\documentclass{article}

\usepackage{arxiv}

\usepackage[T1]{fontenc}
\usepackage[utf8]{inputenc}
\usepackage{tocloft,amsmath,amsfonts,amssymb,amsthm,mathtools,dsfont,bbold,graphicx,fancybox,enumitem,hyperref}

\usepackage{tikz}
\usetikzlibrary{quotes,angles}

\newcommand{\RR}{\mathbb{R}} 
\def\one{\mathbb 1}
\def\tilde{\widetilde}

\def\b#1{\boldsymbol{#1}}

\def\Diag{\mathrm{Diag}}

\newtheorem{theorem}{Theorem}[section]

\usepackage{algorithmic}
\usepackage[ruled,linesnumbered]{algorithm2e}

\usepackage{caption}
\usepackage{subcaption}
\usepackage{booktabs}       

\title{Node and Edge Nonlinear Eigenvector Centrality for Hypergraphs}

\author{%
  \textbf{Francesco Tudisco}
   \\
  School of Mathematics\\
  Gran Sasso Science Institute\\
  67100, L'Aquila (Italy) \\
  \texttt{francesco.tudisco@gssi.it} 
   \And
   \textbf{Desmond J. Higham} \\
   School of Mathematics\\
   University of Edinburgh\\
   EH93FD, Edinburgh (UK)\\
   \texttt{d.j.higham@ed.ac.uk} 
}

\begin{document}

\maketitle

\begin{abstract}
Network 
scientists have shown that there is great value in studying pairwise interactions between components in a system. From a linear algebra point of view, 
this involves defining and evaluating functions of the associated adjacency matrix.
Recent work indicates that 
there are further benefits from 
accounting 
directly for higher order interactions, notably through a hypergraph representation 
where an edge may involve multiple nodes.
Building on these ideas, we motivate, define and analyze a class of 
spectral 
centrality measures for
identifying important nodes and hyperedges in hypergraphs, generalizing existing
network science concepts.   
By exploiting the latest developments in nonlinear Perron-Frobenius theory, 
we show how the resulting constrained nonlinear eigenvalue problems have unique
solutions that can be computed efficiently via a nonlinear
power method iteration.
We illustrate the 
measures on realistic data sets. 
\end{abstract}

\section{Introduction}\label{sec:mot}

The study of pairwise interactions has 
led to a  vast range of useful concepts and tools in network science 
\cite{EstradaKnight2015book,Newman2010book}.
Several recent studies have developed extensions that account for higher order interactions \cite{battiston2020beyond,estrada2016hypergraph,torres2020why}. 
Of course, the appropriate higher order representation 
is dependent on the research problem being addressed.
For example,
as discussed in 
\cite{torres2020why},
 in studying coauthorship data set one could pose three distinct questions:

\begin{enumerate}
\item[1.] Have two given authors ever contributed simultaneously to a multi-authored paper? A simple undirected graph records such pairwise interactions.

\item[2.] Has a given set of authors 
ever contributed simultaneously to a multi-authored paper?
Because any subset of these authors must also have 
contributed simultaneously to a multi-authored paper, we may use 
a simplicial complex to record these interactions.
This structure incorporates downward closure: any subset of nodes within a simplex
also forms a simplex.

\item[3.] Has a given set of authors formed the complete 
coauthorship list on a paper?
In this case, a hypergraph is appropriate, with the set of authors forming a 
hyperedge.
Any proper subset of those authors will not appear 
as a hyperedge unless they form the complete coauthorship list
on some other paper.
\end{enumerate}

In this work we are concerned with extensions of so-called eigenvector centrality measures to the higher-order setting of hypergraphs. 
 Eigenvector centrality for graphs  has been widely used to assign levels of importance to  individual nodes. It assigns scores to nodes in terms of the Perron eigenvector of the adjacency matrix of the graph
 \cite{bonacich1987eig,EstradaKnight2015book}. A standard approach when dealing with hypergraphs is to use graph-based algorithms on the clique-expanded graph of the hypergraph \cite{agarwal2006higher}. In this case, we can assign centrality scores to the nodes of the hypergraph by looking at the centrality of the nodes in its clique-expanded graph.
Another relatively well-established idea to represent and work with hypergraphs  relies on the use of higher-order tensors \cite{arrigo2020framework,benson2019three}. This approach requires a uniform hypergraph,
and defines 
centrality scores in terms of the Perron eigenvector of the adjacency tensor. 

We note that the clique-expansion matrix approach is a form of flattening which essentially corresponds to an additive model: 
for a given node, the importances of its neighbors in a hyperedge are summed into a linear (possibly weighted) combination. Instead, the  tensor-eigenvector approach  is a multiplicative model: the importances of the neighbors in each hyperedge are multiplied. While the two models coincide on standard graphs (as each hyperedge involves exactly two nodes, and hence there is at most one neighbor),  these two models are intrinsically different on uniform hypergraphs with larger hyperedges. See the \textit{Results and Discussion} section for more details.

In this work we define a general eigenvector model for node and edge centralities on hypergraphs with arbitrary hyperedge size based on the hypergraph incidence matrix and the choice of four nonlinear functions. Working in terms of the  incidence matrix
provides a general yet simple model which, for example,  immediately transfers to the case of   simplicial complexes, where all subsets of each hyperedge happen to be present. 
The choice of the nonlinear functions allows us to specify the  way node importances are combined within the hyperedges (and, vice-versa, the way hyperedge importances influence node scores). We  show that both the clique-expansion and, for uniform hypergraphs, the tensor-based eigenvector centralities are particular cases of the proposed model obtained for specific choices of the nonlinearities.  Thus, our approach allows us to generate a whole new family of centrality models and to extend popular tensor-eigenvector centrality models to general hypergraphs in a natural way. More precisely, the main contributions of this work as follows.

We formulate node and edge hypergraph centrality as a general constrained nonlinear eigenvalue problem \eqref{eq:NEP} based on the hypergraph incidence operator. We provide existence and uniqueness theory for the eigenvalue equation in Theorem~\ref{thm:existence_uniqueness}  and we propose a Nonlinear Power Method (NPM) to compute its solution.  In the special case of a 2-uniform hypergraph (a graph) this leads to a Hypertext Induced Topic Search (HITS) type iteration for networks that simultaneously    assigns centrality to nodes and edges, whereas, for general  hypergraphs, the NPM allows us to compute centrality scores for  nodes and general hyperedges.  The convergence theory for the NPM  is provided in Theorem~\ref{thm:convergence}. Finally, in the \textit{Experiments} section \S\ref{sec:comp}, we provide several computational experiments on synthetic and real-world data to illustrate the behavior of different  node and edge centrality models obtained as particular cases of the general constrained nonlinear eigenvalue equation \eqref{eq:NEP}.   

Our main contributions are presented in the  \textit{Results and Discussion} section \S\ref{sec:results} while all proofs 
appear in the \textit{Methods} section \S\ref{sec:methods}.

\section{Results and Discussion}\label{sec:results}

\subsection{Notation and Motivation}\label{sec:not}

We begin by 
 considering an unweighted, undirected, graph 
$G=(V,E)$ with node set $V=\{1,\dots,n\}$, edge set $E = \{e_1,\dots,e_m\}$
and binary adjacency matrix $A  \in \RR^{n \times n}$.
In this context, eigenvector 
centrality was popularized in the social network science
community
\cite{bonacich1987eig}, although the idea can be traced back to the 19th Century
\cite{schaf2019landau}. This centrality model assigns a measure of importance, $x_i > 0$, to node $i$
in such a way that the importance of
node $i$ is linearly proportional to the sum of the importances of its neighbours. 
This relationship may be written
\begin{equation}\label{eq:eig-centrality-A}
  A\b x = \lambda \b x, \quad \b x>0, \quad \text{for~some~} \lambda >0.
\end{equation}
 Thanks to the Perron--Frobenius theorem, this matrix eigenvector problem admits a unique solution $\b x^*$ if $A$ is irreducible (that is, the graph is connected)
\cite{EstradaKnight2015book}.
 In this case, $\b x^*$ can be computed to arbitrary precision via the 
 power method if $A$ is primitive (that is, the graph is aperiodic), see, 
 for example, \cite{tudisco2015complex}.

Over the years,  a large amount of work has been devoted to the definition of  centrality models able to capture different network properties and thus provide an importance score to the nodes of a graph. However, much less work has focused on models and methods for quantifying centrality of edges. 
In addition to being of interest in its own right, quantifying edge importance has natural applications in a number of important tasks, including link detection, edge prediction and matrix completion \cite{gleich2016seeded,kim2011link,dhillon2012link}.

An eigenvector centrality for edges can be developed by considering the line graph and its adjacency matrix $A^{(e)}$ \cite{brohl2019edge}. In this setting, $A^{(e)} \in \RR^{m \times m}$ has  
$A^{(e)}_{e_1,e_2}\neq 0$ if and only if $e_1\in E$ and $e_2\in E$ share at least one node. The centrality of the edges can thus be defined via the Perron eigenvector $A^{(e)}\b y = \lambda \b y$, as in \eqref{eq:eig-centrality-A}.

Another somewhat natural model for edge centrality can be induced by a given node centrality $\b x$: assign to each edge $e$ the score $y_e = \sum_{i\in e}x_i$ obtained by looking at the nodes the edge connects.
This is what is typically done when computing edge scores for, e.g., link prediction \cite{kim2011link,cipolla2021nonlocal,liben2007link} or network robustness optimization \cite{chan2014make,arrigo2016updating}. However, a mutually reinforcing centrality score for both nodes and edges can be designed by requiring the edges to inherit importance from the nodes they connect and, vice-versa, the nodes to inherit importance from the edges they participate in.
We describe this idea in detail in the next section. 
Since the resulting model extends essentially unchanged to the higher-order setting where edges contain an arbitrary number of nodes, we will present this idea in the framework of a hypergraph.

\subsection{Hypergraphs}
\label{sec:hyp}
In the remainder of this work 
we consider a general hypergraph $H=(V,E)$ where $V=\{1,\dots,n\}$ is the set of nodes and $E = \{e_1,\dots,e_m\}$ now denotes the set of hyperedges
\cite{ouvrad20}.
Note that in our setting every node can belong to an arbitrary number of hyperedges, i.e., we allow hereditary hypergraphs \cite{bretto2013hypergraph}, which can also be used to model a simplicial complex structure. We let $B$ denote the $n\times m$ incidence matrix of $H$, defined as follows: the rows of $B$ correspond to nodes while its columns correspond to hyperedges and we have $B_{i,e}=1$ if node $i$ takes part in hyperedge $e$, that is,  
\[
B_{i,e} = 
\begin{cases}
1 & i\in e \\
0 & \text{otherwise.}
\end{cases}
\]

In many situations we have access to external node and edge weights in the form of weight functions $\nu:\RR^V\to\RR_+$ and $w:\RR^E\to\RR_+$. For example, if the hypergraph data represents grocery goods (the nodes) and the list of items purchased by each customer in one visit to a supermarket (the hyperedges), then $\nu(i)$ can be an indicator of the price of item $i$, while $w(e)$ may correspond to the profit the supermarket has made with the list of items in $e$. We use two diagonal matrices $N$ and $W$ to take into account for these weights, defined by $N = \Diag(\nu(1), \dots, \nu(n))$ and 
$W = \Diag(w(e_1),\dots,w(e_m))$.

Note that a hypergraph where all edges have exactly two nodes is a standard graph. In that case, 
we have $BWB^\top = A + D$ where $A$ is the adjacency matrix of the graph and $D = \Diag(d_1, \dots, d_n)$ is the digonal matrix of the weighted node degrees $d_i = \sum_{e\in E}w(e)B_{i,e} = (BW\one)_i$.
Similarly, for a general hypergraph $H$, we have $BWB^\top = A_H+D_H$ where $A_H$ and $D_H$ are the adjacency and degree matrices of the clique-expansion graph $G_H=(V,E_H)$ associated with $H$. The clique-expansion graph is a graph on the same vertex set of $H$, obtained by adding a weighted clique connecting all nodes in each hyperedge of $H$. More precisely, given $H=(V,E)$,  we have 
$$
(A_H)_{ij} = \sum_{e: \, i,j\in e}w(e) , \qquad (A_H)_{ii}=0\, .
$$
Thus, $ij\in E_H$ if and only if $i\neq j$ participates in at least one hyperedge of $H$. Similarly, the degree matrix $D_H=\Diag(d_1, \dots, d_n)$ is the diagonal matrix whose diagonal entries are the weighted degrees of the nodes in the hypergraph, i.e.\ $d_i = \sum_{e:i\in e}w(e)=(BW\one)_i$.

\subsection{Node and edge nonlinear hypergraph eigenvector centrality}\label{sec:centrality-hypergraph}

We describe here a spectral (thus mutually reinforcing) model for node and edge centralities of hypergraphs. Suppose $H = (V,E)$ is given with $|V|=n$ and $|E|=m$, and let $\b x \in \RR^n$, $\b y \in \RR^m$ be nonnegative vectors whose entries will provide centrality scores for the nodes and 
hyperedges of $H$, respectively. We would like the importance $y_e$ for an edge $e\in E$ to be a nonnegative number proportional to a function of  the importances of the nodes in $e$, for example  $y_e \propto \sum_{i\in e}\nu(i)x_i$. Similarly, we require that the centrality $x_i$ of node $i\in V$ is a nonnegative number proportional 
to a function of the importances of the edges it participates in, for example $x_i \propto \sum_{e: i\in e}w(e)y_e$. As the centralities $x_i$ and $y_e$ are all nonnegative, these sums coincide with the weighted $\ell^1$ norm of specific sets of centrality scores. Thus, we can generalize this idea by considering the weighted $\ell^p$ norm of node and edge importances. This leads to 
\[
x_i \propto \Big(\sum_{e: i\in e}w(e)y_e^p\Big)^{1/p},\qquad y_e \propto \Big(\sum_{i\in e}\nu(i)x_i^q\Big)^{1/q},  
\]
for some $p,q\geq 1$. More generally, we can consider four functions $f,g,\varphi,\psi:\RR_+\to\RR_+$ of the nonnegative real line $\RR_+$ and require that 
\[
x_i \propto g\Big(\sum_{e: i\in e}w(e)f(y_e)\Big),\qquad y_e \propto \psi\Big(\sum_{i\in e}\nu(i)\varphi(x_i)\Big)  \, .
\]
If we extend real functions on vectors by defining them as mappings that act in a componentwise fashion, the previous relations can be compactly written as the following constrained nonlinear  equations
\begin{equation}\label{eq:NEP}
    \begin{cases}
    \lambda \b x = g\big(BW f(\b y)\big) & \\
    \mu \b y = \psi\big(B^\top N \varphi(\b x)\big)
    \end{cases}\qquad \b x,\b y > 0, \quad  \lambda, \mu > 0  \, .
\end{equation}

If $f,g,\psi$ and $\varphi$ are all identity functions, then \eqref{eq:NEP} boils down to a linear system of equations which is structurally reminiscent of the 
HITS centrality algorithm for directed graphs, \cite{arrigo2019multi,k99}.
HITS computes two different node centralities: a \textit{hub} centrality, which is proportional to the authority score of neighboring nodes, and at the same time, \textit{authority} centrality, which is proportional to the hub score of neighboring nodes. Similarly, \eqref{eq:NEP} with $f=g=\varphi=\psi=\text{id}$ defines two centralities, but in this case
they relate to nodes and hyperedges: the importance of a node is proportional to the sum of the importances of the hyperedges it belongs to and, vice-versa, the importancesof a hyperdge is proportional 
to the sum of the importances of the nodes it involves.

As for HITS centrality, when $f=g=\varphi=\psi=\text{id}$ and we have no edge nor node weights (i.e.\ $W,N$ are identity matrices), then  $\b x,\b y$ in  \eqref{eq:NEP} are the left and right  singular vectors of a certain matrix,  in this case $B$, and the matrix Perron--Frobenius theory tells us that if the bipartite graph with adjacency matrix 
$
\begin{psmallmatrix}
0 & B\\
B^\top & 0
\end{psmallmatrix}
$
is connected, then  \eqref{eq:NEP} has a unique solution. Instead, when either $f,g,\varphi$ or $\psi$ is not linear, even the most basic question of existence of a solution to \eqref{eq:NEP} is not straightforward. However, for homogeneous functions $f,g,\varphi$ and $\psi$, the nonlinear Perron--Frobenius theory for multihomogeneous operators \cite{gautier2019perron} allows us to give guarantees on existence, uniqueness and computability for the nonlinear singular-vector centrality model in \eqref{eq:NEP}.

Before addressing these issues, we  
investigate the system in \eqref{eq:NEP} further,
showing how it includes some previously proposed eigenvector centrality models as special cases, and 
offers additional useful flexibility.

\subsection{The linear case: eigenvector centrality for graph and line graph}\label{sec:linear-case}
When $H$ is a standard simple and unweighted graph $H=G=(V,E)$, with binary adjacency matrix $A$, it is easy to verify that  $BB^\top = A + D$, where $D$ is the diagonal matrix of the node degrees. Moreover, $B^\top B = A^{(e)} + \Delta$, where $A^{(e)}$ is the adjacency matrix of the line graph of $G$ and $\Delta = \Diag(\delta_1,\dots,\delta_m)$ is a diagonal matrix whose diagonal entries count the number of nodes each edge contains. In this case, each edge has exactly two nodes, so $\Delta = 2I$. 
The corresponding identities hold if we allow weights on the nodes and on the edges of $G$, namely $BWB^\top = A + D$, where now $A$ and $D$ are the weighted adjacency and degree matrix of $G$, and $B^\top NB = A^{(e)} + \Delta$ with 
\begin{equation}\label{eq:adj-mx-line-graph-of-G}
    (A^{(e)})_{e_1,e_2} = \begin{cases}\nu(i) & e_1\neq e_2 \text{ and they share the node }i\text{ in }G\\
0 & \text{otherwise}
\end{cases}
\end{equation}
and $\delta_{e}= \sum_{i\in e}\nu(i)$.

It follows that, when $H$ is a graph, the node-edge eigenvector model in \eqref{eq:NEP} for the linear case $f=g=\varphi=\psi=\text{id}$ is strongly related to the standard eigenvector centrality applied to $G$ and its line graph. In fact, by using the two identities $\lambda \b x = BW\b y$ and $\mu \b y = B^\top N \b x$ we obtain  
$$
\begin{cases}
\tilde \lambda\,  \b x = BWB^\top N\b x = (A+D)N\b x & \\
\tilde \lambda\, \b y = B^\top N BW \b y = (A^{(e)}+\Delta)W \b y
\end{cases}
$$
with $\tilde \lambda = \lambda \mu$. Thus, $\b x$ and $\b y$ are the Perron eigenvectors of diagonally perturbed adjacency matrices of the graph and the line graph.  

A similar connection holds for the general hypergraph case. In that case, the node-edge eigenvector model in \eqref{eq:NEP} for the linear choices $f=g=\varphi=\psi=\text{id}$ is tightly connected to the eigenvector centrality of the clique-expansion graph of $H$ its line graph. Precisely we have
$$
BWB^\top = A_H + D_H \qquad \text{and}\qquad B^\top N B = A^{(e)}_H + \Delta_H, 
$$
where $A^{(e)}_H$ and $\Delta_H$ are the adjacency and degree matrix of the line graph of $G_H$, as defined in~\eqref{eq:adj-mx-line-graph-of-G}.

\subsection{Tensor-based eigenvector centrality for uniform hypergraphs and its extension} 
\label{sec:tensor-centrality}
In this subsection we find an intriguing connection between recently proposed tensor-based eigenvector centralities for uniform hypergraphs  and the  nonlinear singular vector model proposed in  \eqref{eq:NEP}. In particular, we show that the centrality models based on tensor eigenvectors are a special case of our general nonlinear singular vector framework and that this 
new approach allows us to extend tensor eigenvector centralities to general non-uniform hypergraph data. We first review the uniform hypergraph case.

A hypergraph is said to be 
$k$-uniform if $|e|=k$ for all $e\in E$. 
Thus, a $2$-uniform hypergraph is a graph in the standard sense.
The concept of eigenvector centrality has been extended to the case of $k$-uniform hypergraphs with $k>2$ by means of the hypergraph adjacency tensor \cite{benson2019three}. As every hyperedge contains exactly $k$ nodes, we can associate to $H$ a tensor 
$\mathcal A$ 
with $k$ indices $\mathcal A_{i_1,\dots,i_k}$ such that $\mathcal A_{i_1,\dots,i_k} = w(e)$ if the hyperedge $e = \{i_1,\dots,i_k\}$ belongs to $E$, and $\mathcal A_{i_1,\dots,i_k}=0$ otherwise. Clearly, $\mathcal A$ coincides with the adjacency matrix of the graph when $k=2$. Different notions of tensor eigenvectors are available in the literature (see e.g.,
 \cite{cipolla2019shifted,gautier2019unifying}). In particular, for $p>0$, a $\ell^p$  eigenvector for $\mathcal A$ is a vector $\b x$ such that 
\begin{equation}\label{eq:tensor_eig}
    \sum_{i_2,\dots,i_k}\mathcal A_{i_1,i_2,\dots,i_k}x_{i_2}x_{i_3}\cdots x_{i_k} = \lambda \, x_{i_1}^{\, p} \, .
\end{equation}
The special cases $p=1$ and $p={k-1}$ correspond to so-called $Z$- and $H$-eigenvectors for $\mathcal A$. Note that both $Z$- and $H$-eigenvectors boil down to standard matrix eigenvectors when $k=2$. However, when $k>2$ matters are significantly different. In particular, the eigenvector centrality defined by \eqref{eq:tensor_eig} is no longer linear when $k>2$, in the sense that 
taking a linear combination of 
eigenvectors does not automatically produce an eigenvector.

This nonlinearity makes the analysis and the computation of solutions to \eqref{eq:tensor_eig} more challenging than standard eigenvector centralities for graphs (i.e., $k=2$). However, it has been observed in, e.g., 
\cite{gautier2019contractivity,tudisco2017node,gautier2019unifying} that the nonlinear eigenvector equation \eqref{eq:tensor_eig} admits a unique solution 
that can be computed to an arbitrary precision if the tensor $\mathcal A$ is not too sparse and if the exponent $p$ satisfies certain assumptions. In particular, for a large range of values of $p$ all these properties hold with almost no requirement on the connectivity of the underlying hypergraph \cite{cipolla2019shifted}. From this point of view, the nonlinearity yields a remarkable advantage rather than a disadvantage.

Extending tensor eigenvector centrality models to the non-uniform hypergraph setting is not straightfoward. The next theorem shows that our nonlinear singular vector model in \eqref{eq:NEP} provides a natural framework to this end. In fact, Theorem \ref{thm:tensor-eig} shows that, for uniform hypergraphs,  the tensor-based eigenvector centrality in \eqref{eq:tensor_eig} is a particular case of \eqref{eq:NEP} for logarithmic- and exponential-based nonlinear functions. Thus, when used on non-uniform hypergraphs, these choices of functions in \eqref{eq:NEP} yield a tensor eigenvector like centrality for general hypergraphs. We will further discuss this extension in the \textit{Experiments} section \S\ref{sec:comp}. 

Let $H$ be a  $k$-uniform hypergraph.
As observed above, when $k=2$ we have $BWB^\top = A + D$, where $A$ and $D$ are the adjacency and degree matrices of the graph $H$, respectively. Thus, for $p=2$ we can easily rewrite the eigenvector centrality equation  \eqref{eq:tensor_eig} in terms of the incidence matrix, as   \eqref{eq:tensor_eig} coincides with $A\b x = \lambda \b x$ and we have $A\b x= (BWB^\top - D)\b x = \lambda \b x$.  If the vector $\b x$ is entrywise positive, we can add a nonlinear transformation in the eigenvector equation to obtain a similar relation for any $k\geq 2$ and any $p\geq 1$. More precisely, we have
\begin{theorem}\label{thm:tensor-eig}
Let $H$ be a $k$-uniform hypergraph with $\nu(i) = 1$ for all $i\in V$. If $\b x$ is a positive solution of  \eqref{eq:NEP} with $f(\b x) = \b x$, $g(\b x) =\b x^{1/(p+1)}$,  $\psi(\b x) = e^{\b x}$ and $\varphi(\b x) = \ln(\b x)$,  then $\b x$ is an eigenvector centrality solution of the tensor eigenvalue problem~in \eqref{eq:tensor_eig}. 
\end{theorem}

\subsection{Existence, uniqueness and computation of nonlinear hypergraph centralities}\label{sec:theory}

In this section we discuss existence, uniqueness, positivity, maximality and computation of the node and edge hypergraph centrality defined by the general nonlinear singular value problem in \eqref{eq:NEP}. Analogously to the linear case, these properties will follow directly from the nonlinear Perron--Frobenius theorem for multihomogeneous mappings \cite{gautier2019perron}, which extends the classical Perron--Frobenius theory for nonnegative matrices to a much broader class of nonlinear nonnegative operators. 

 To this end, we recall that a function $\varphi$ is said to be $\alpha$-homogeneous if $\varphi(\lambda \b u) = \lambda^\alpha \varphi(\b u)$ for all $\lambda\geq 0$. In this case we say that $\alpha$ is the homogeneity degree of $\varphi$. Furthermore, we say that $\varphi$ is  order preserving if $\varphi(\b v) \geq \varphi(\b u)$ for all $\b v\geq \b u$; whereas we say that $\varphi$ is  positive if $\varphi(\b v) > 0$ for all $\b v> 0$. We have
\begin{theorem}\label{thm:existence_uniqueness}
 Let $f,g,\varphi,\psi$ be order preserving and homogeneous of degrees $\alpha,\beta,\gamma,\delta$, respectively. Define the coefficient $\rho = |\alpha\beta\gamma\delta|$. If either
 \begin{itemize}
     \item[P1.] $\rho<1$, or
     \item[P2.] $\rho=1$;  $f,g,\varphi,\psi$ are differentiable and positive maps;  the bipartite graph with adjacency matrix 
$
\begin{psmallmatrix}
0 & BW\\
B^\top N & 0
\end{psmallmatrix}
$
is connected
 \end{itemize}
 then there exist unique $\b x,\b y > 0$ (up to scaling) and unique $\lambda, \mu >0$ solution of \eqref{eq:NEP}.  
\end{theorem}

By analogy with the linear case, we refer to the positive solutions of \eqref{eq:NEP} defining the hypergraph centralities as nonlinear Perron singular vectors. 

\begin{algorithm}[t]
  \DontPrintSemicolon
	\caption{Nonlinear Power Method for hypergraph centrality}\label{alg:npm}
    \KwIn{Incidence matrix $B$ of the hypergraph; diagonal weight matrices $W$ and $N$ for edges and nodes; nonlinear functions $f,g,\varphi,\psi$; desired vector norm $\|\cdot\|$; stopping tolerance $tol$ } 
	\KwOut{Centrality for nodes $\b x$ and hyperedges $\b y$ such that $\|\b x\|=\|\b y\|=1$}
      $\b x^{(0)}, \b y^{(0)} >0$  \qquad \texttt{\# Initialize with any positive vectors}\;
	 \Repeat{$\|\b x^{(r+1)}-\b x^{(r)}\|/\|\b x^{(r+1)}\| +  \|\b y^{(r+1)}-\b y^{(r)}\|/\|\b y^{(r+1)}\|< tol$}{
	   $\b u \gets \sqrt{\b x^{(r)}\, g\big(BWf(\b y^{(r)})\big)}$\qquad\,\, \texttt{\#entrywise multiplication and squareroot}\;
	   $\b v \gets \sqrt{\b y^{(r)}\, \psi\big(B^\top N\varphi(\b x^{(r)})\big)}$\qquad \texttt{\#entrywise multiplication and squareroot}\;	   $\b x^{(r+1)}\gets \b u / \|\b u\|$\;
	   $\b y^{(r+1)}\gets \b v/\|\b v\|$\; 
	   }
\end{algorithm}

On top of existence and uniqueness guarantees, the matrix Perron--Frobenius theorem provides us with the convergence of the so-called power method, a very powerful  tool for computing the Perron singular vectors. 
In the case of nonnegative matrices, however,  one needs to require the bipartite graph of $A=\begin{psmallmatrix}
0 & BW\\
B^\top N & 0
\end{psmallmatrix}$ to be aperiodic (i.e., the matrix $A$ is primitive) in order to ensure the convergence of the power method for an arbitrary choice of the starting point, as connectedness alone is not enough. As an example, consider the $2\times 2$ matrix $A = \begin{psmallmatrix}
0 & 1\\
1 & 0
\end{psmallmatrix}$ which acts on vectors by swapping the first and the second coordinates. The graph of $A$ is connected but not aperiodic and, indeed, the sequence $\b x^{(r+1)} = A\b x^{(r)}$ does not converge in general (it converges only if $\b x^{(0)}$ has constant entries). 

Much like the matrix case, one can compute the nonlinear Perron singular vectors in \eqref{eq:NEP} via what we call \textit{Nonlinear Power Method}, described in Algorithm \ref{alg:npm}. However, similarly to the nonlinear eigenvalue problem for tensors \cite{cipolla2019shifted,gautier2019unifying}	, the nonlinear power method converges under significantly milder conditions than its more common linear counter part. In particular, no aperiodicity assumption is required and, depending on the homogeneity degree of the nonlinear functions, global convergence may be ensured even for disconnected graphs.   The following theorem describes the convergence of Algorithm \ref{alg:npm} to the solution of \eqref{eq:NEP}.

\begin{theorem}\label{thm:convergence}
Under the assumptions and notation of Theorem \ref{thm:existence_uniqueness}, let $\b x^{(r)}$, $\b y^{(r)}$ be the sequences generated by Algorithm \ref{alg:npm}. If either P1 or P2 holds, 
 then $\b x^{(r)}$ and $\b y^{(r)}$ converge to the unique positive solutions $\b x^*,\b y^*$ of \eqref{eq:NEP} such that $\|\b x^*\|=\|\b y^*\|=1$. Moreover,  if P1 holds, then the convergence is linear, i.e.:
 $$
 \|\b x^{(r)}-\b x^*\|+\|\b y^{(r)}-\b y^*\| = O( \rho^r).
 $$
\end{theorem}

\subsection{Experiments} \label{sec:comp}

In this section we compare  the behaviour of three node-edge eigenvector centrality models which correspond to three choices of the functions $f,g,\varphi,\psi$ in \eqref{eq:NEP}, as described below:

\begin{itemize}
\item The ``linear'' centrality model corresponds to the choice $f=g=\varphi=\psi=\mathrm{id}$ which, as discussed in the \textit{Node and edge nonlinear hypergraph eigenvector centrality} section \S\ref{sec:centrality-hypergraph}, 
essentially corresponds to the standard eigenvector centrality applied to the graph and the line graph obtained by clique-expanding the input hypergraph. 
\item The ``log-exp'' centrality model  corresponds to the choices $f=\mathrm{id}$, $\varphi(x) = \ln (x)$, $\psi(x) = \exp(x)$ and $g(x) = \sqrt{x}$. As discussed in the \textit{Tensor-based eigenvector centrality for uniform hypergraphs and its extension} section \S\ref{sec:tensor-centrality}, 
this choice generalizes the tensor eigenvector centrality proposed in \cite{benson2019three}
to the case of nonuniform hypergraphs.
In fact, when the hypergraph is uniform, we have already observed in Theorem \ref{thm:tensor-eig} that the node centrality defined via \eqref{eq:NEP} with this choice of $f,g,\varphi,\psi$ boils down to a $Z$-eigenvector of the adjacency tensor of the hypergraph. Similarly, for a general hypergraph, from \eqref{eq:NEP} we get
$$
\mu y_e = \psi(\sum_{j\in e} \nu(j) \varphi(x_j)) = \exp(\sum_{j\in e} \nu(j) \ln(x_j)) = \prod_{j\in e}x_j^{\nu(j)}.
$$
Thus, if $\b x, \b y$ are  nonnegative vectors satisfying \eqref{eq:NEP} there exists a positive $\tilde \lambda$ such that 
\begin{equation}\label{eq:log-exp-node}
    \tilde \lambda x_i^2 =  x_i^{\nu(i)} \sum_{e : i \in e} w(e) \prod_{j\in e\setminus \{i\}}x_j^{\nu(j)} \, .
\end{equation}
Note in particular that, when the input hypergraph has binary node weights, i.e., $\nu(i) = 1$ for all $i\in V$,  \eqref{eq:log-exp-node}  corresponds to a nonuniform hypergraph version of the tensor Z-eigenvector centrality for uniform hypergraphs, precisely we have 
$$
\tilde \lambda x_i =   \sum_{e : i \in e} w(e) \prod_{j\in e\setminus \{i\}}x_j\, .
$$
\item The ``max'' centrality model is based on the obseravtion that the function $\mu_\alpha(\b v) = (v_1^\alpha+\cdots + v_m^\alpha)^{1/\alpha}$ is a
type of `softmax': when $\alpha \to \infty$,  $\mu_\alpha(\b v)$ converges to $\max(\b  v) = \max\{v_1,\dots, v_m\}$. More precisely, we have 
$$
\max(\b v) \leq (v_1^\alpha+\cdots + v_m^\alpha)^{1/\alpha} \leq m^{1/\alpha} \max(\b v)\, ,
$$
thus, already for $\alpha = 10$ we have $\mu_\alpha(\b v) \approx \max(\b v)$. 

Based on this observation, the proposed centrality model corresponds to the choice of  nonlinear mappings:
$f=g=\text{id}$, $\varphi(x)= x^\alpha$ and  $ \psi(x) = \varphi^{-1}(x) = x^{1/\alpha}$, for $\alpha = 10$. 
Notice that,  with this choice of $\alpha$, the max node centrality $x_i$ is large when $i$ is part of at least one important edge. 
In fact, from \eqref{eq:NEP} we have 
$$
x_i \approx \max\{y_e : e \supset i\}\, .
$$

\end{itemize}

\subsubsection{Illustrative example: hypergraph sunflower}
A sunflower is a hypergraph whose hyperedges all have one common intersection in one single node, called the \textit{core}. Let $u\in V$ be that intersection. Also let $r$ be the number of \textit{petals} (the hyperedges) each containing $|e_i|$ nodes, for $i=1,\dots,r$. By definition $u\in e_i$ for all $i$.  Further, 
a node $v\in e_i$ and $v\in e_j$ for $i\neq j$ if and only if $v=u$.

\begin{figure}[t]
     \centering
     \begin{subfigure}[b]{0.9\textwidth}
         \centering
        \includegraphics[width=.9\textwidth]{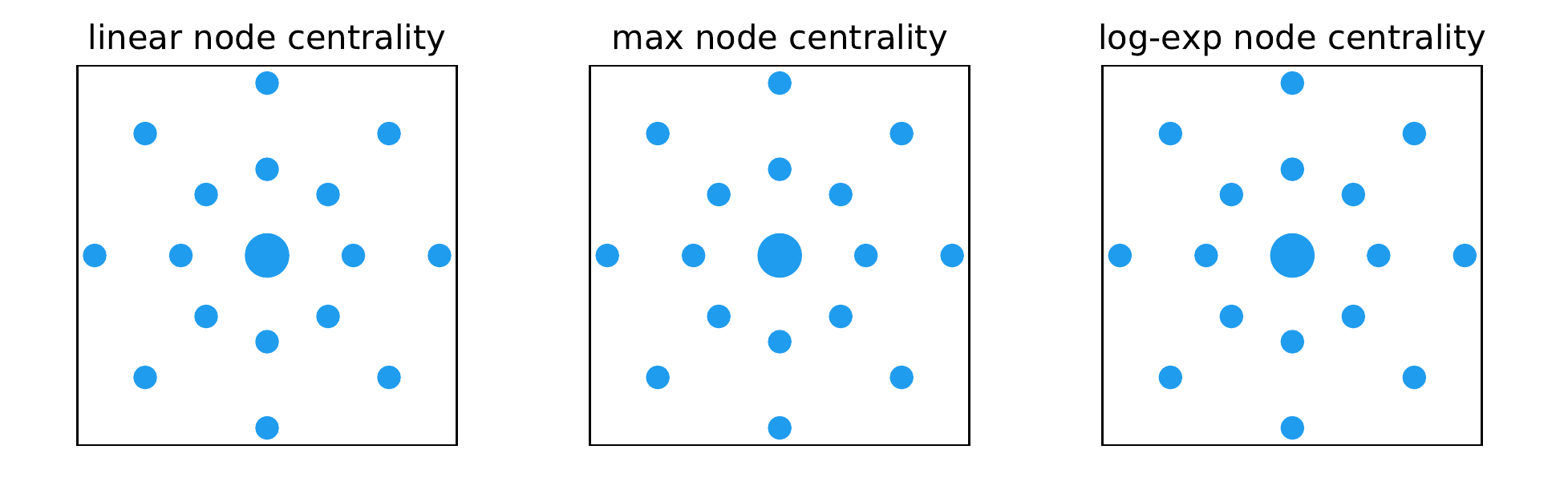}
        \caption{}
        \label{fig:uniform-sunflower}
     \end{subfigure}
     \\
     \begin{subfigure}[b]{0.9\textwidth}
         \centering
        \includegraphics[width=.9\textwidth]{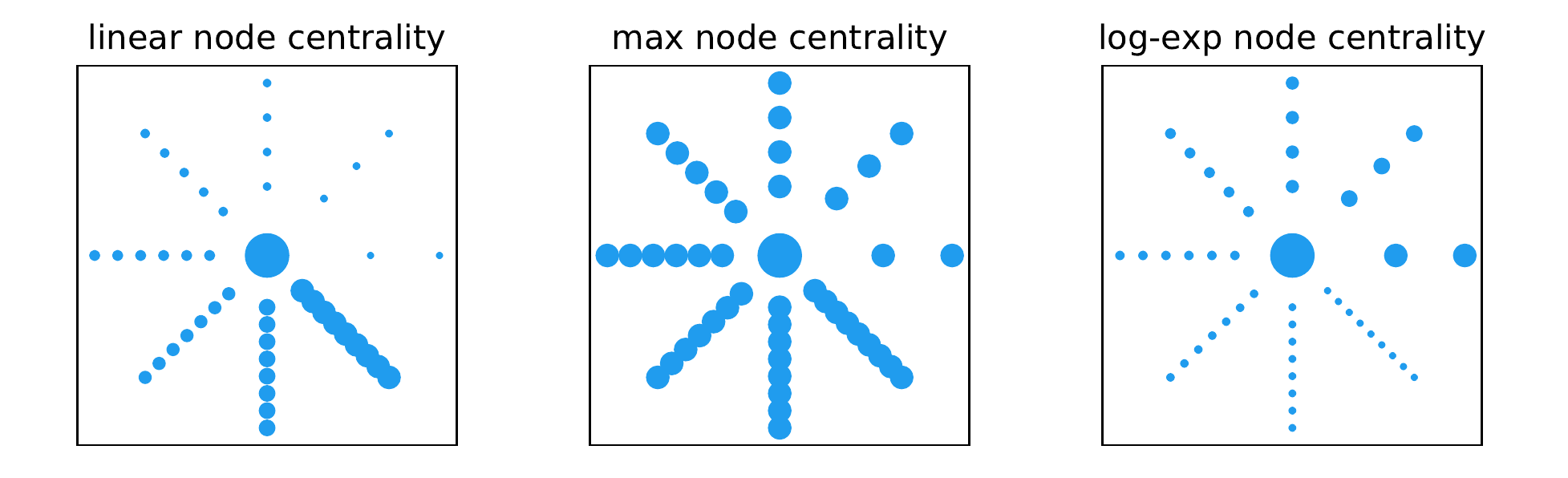}
        \caption{}
        \label{fig:non-uniform-sunflower}
     \end{subfigure}

    \caption{\textbf{Example of node centralities on sunflower hypergraphs.}\\ Node centrality scores for the ``linear'', ``log-exp'' and ``max'' centrality models on the two example hypergraph sunflowers. The corresponding centrality model is specified on top of each panel. Dots represent the hypergraph nodes and their size is proportional to their centrality value. 
    (a) shows results on a uniform sunflower with eight petals each containing three nodes; 
    (b) shows results on a hypergraph sunflower with eight petals containing $3,4,\dots, 10$ nodes, respectively.}
    \label{fig:sunflower}
\end{figure}

Let us first consider the case of a uniform sunflower. This case corresponds to the setting where all the petals have the same number of nodes, i.e.\ $|e_i|=k+1$ for all $i$ and for some integer $k$. 
The tensor eigenvector centrality of a uniform sunflower is studied for example in \cite{benson2019three}. In that case we can assume that all the hyperedges have the same centrality score and that the same holds for all the nodes, besides the core, which is assigned a specific value. 

Assuming no weights on nodes or hyperedges, 
by symmetry we may impose the constraints $x_{v_i}=x_v$ for all $v_i\neq u$ and $y_e = y$ for all $e\in E$
in \eqref{eq:NEP} to obtain
$$
x_v \propto g(f(y)),\qquad  x_u \propto g(rf(y)), \qquad y \propto \psi(\varphi(x_u) + k\varphi(x_v)). 
$$
So, for example, with the choices of Theorem \ref{thm:tensor-eig} we get $x_u/x_v = g(r) = r^{1/(p+1)}$ which coincides with the value computed in \cite{benson2019three}, for the two choices $p = 1$ and $p=m-1$, i.e., the tensor $Z$-eigenvector and $H$-eigenvector based centralities, respectively. More generally, if $g$ is homogeneous of degree $\beta$ we have
\begin{equation}\label{eq:sunflower-centrality-ratio}
    \frac{x_u}{x_v} \propto r^\beta\, .
\end{equation}
This shows that the node centrality assignment in the case of a uniform sunflower hypergraph only depends on the homogeneity degree of $g$ and, in particular,  when $\beta\to 0$ all the centralities tend to coincide, while $x_u > x_v$ for all $\beta>0$,  confirming and  extending the observation in \cite{benson2019three} for the setting of uniform hypergraph centralities based on tensor eigenvectors. Figure \ref{fig:uniform-sunflower} illustrates this behaviour on an example uniform sunflower hypergraph with eight petals ($r=8$) each having three nodes ($k=3$). The figure shows the nodes of the hypergraph with a blue dot whose size is proportional to its centrality value computed according to the ``linear'', ``log-exp'' and ``max'' centrality models. 
The value of $\beta$ for these three centralities is $1$ for both the `max' and the `linear' centrality', and $1/2$ for `log-exp' centrality'. Thus, all the three models assign essentially the same centrality score: the core node $u$ has strictly larger centrality, while all other nodes have same centrality score. Similarly, the computed edge centrality is constant across all models and all petals.

The situation is different for the case of a non-uniform hypergraph sunflower. In this case, we have $r$ petals each containing an arbitrary number of nodes. 
The computational results in Figure~\ref{fig:non-uniform-sunflower} indicate that the ``linear'', ``log-exp'' and ``max'' centrality models 
capture significantly different centrality properties:
All three models recognize the core node as the most central one, however while the `linear' model favours nodes that belong to large hyperedges, the multiplicative `log-exp' model behaves in the opposite way assigning a larger centrality score to nodes being part of small hyperedges. Finally, the `max' model behaves like in the uniform case, assigning the same centrality value to all the nodes in the petals (core node excluded). For this hypergraph, we observe that the edge centrality follows directly from the node one: for the `linear' model the edge centrality is proportional to the number of nodes in the edge, 
for the `log-exp' model it is inversely proportional to the number of nodes, while for the `max' model all edges have the same centrality. 
It would be of interest to pursue these differences analytically and hence gain further 
insights into the effect of $f,g,\varphi$ and $\psi$.

\subsubsection{Real-world hypergraph data}
In this section we analyze the proposed nonlinear node-edge hypergraph centrality model on two real-world datasets. The code used to compute the results of this section is written in \texttt{julia} and is  available at \url{https://github.com/ftudisco/node-edge-hypergraph-centrality}.

\begin{figure}[t!]
    \centering
    \includegraphics[width=\textwidth]{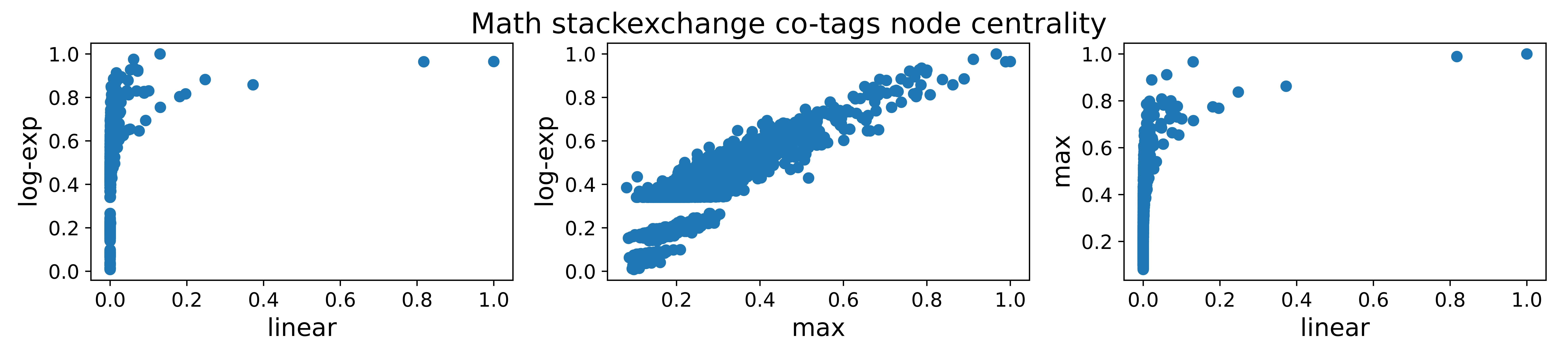}\\
    \includegraphics[width=\textwidth]{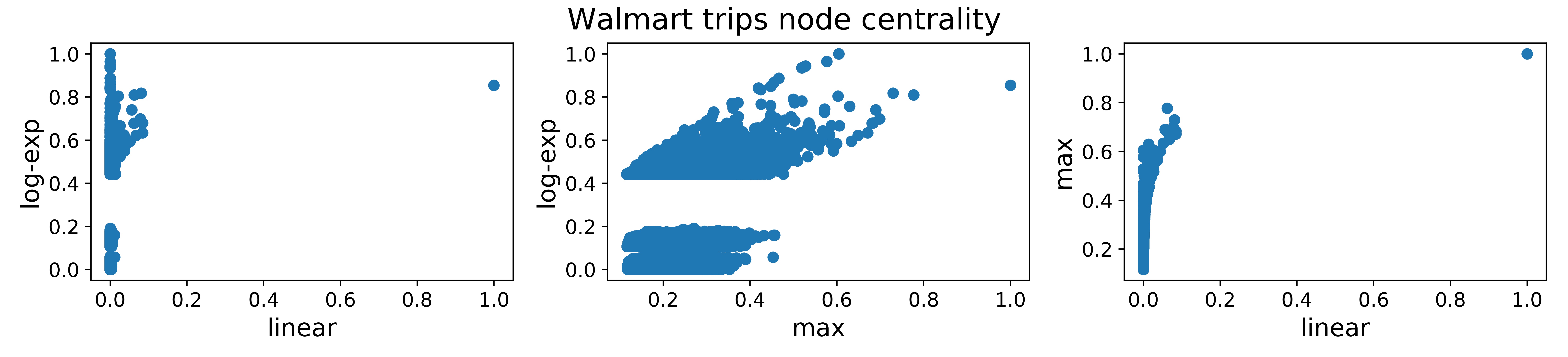}\\
    \includegraphics[width=\textwidth]{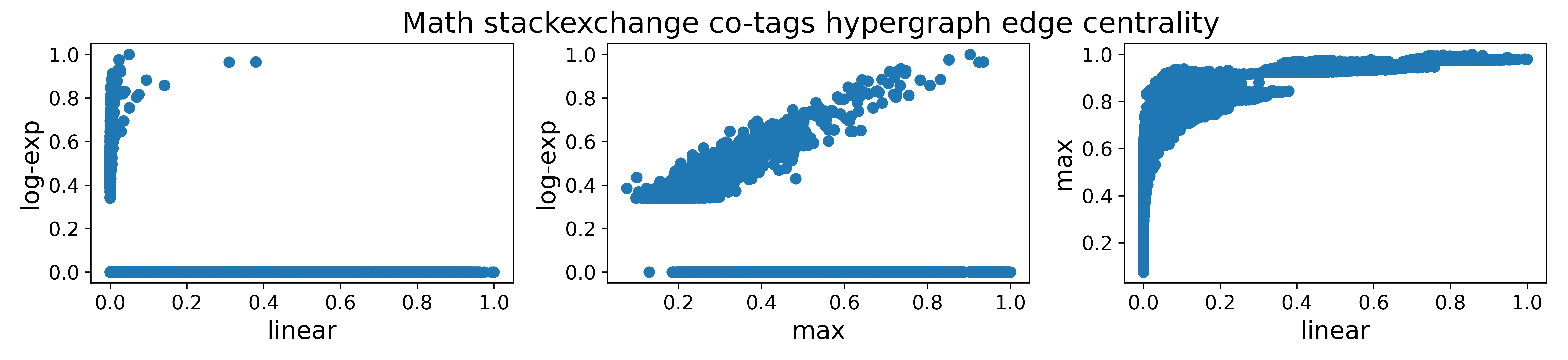}\\
    \includegraphics[width=\textwidth]{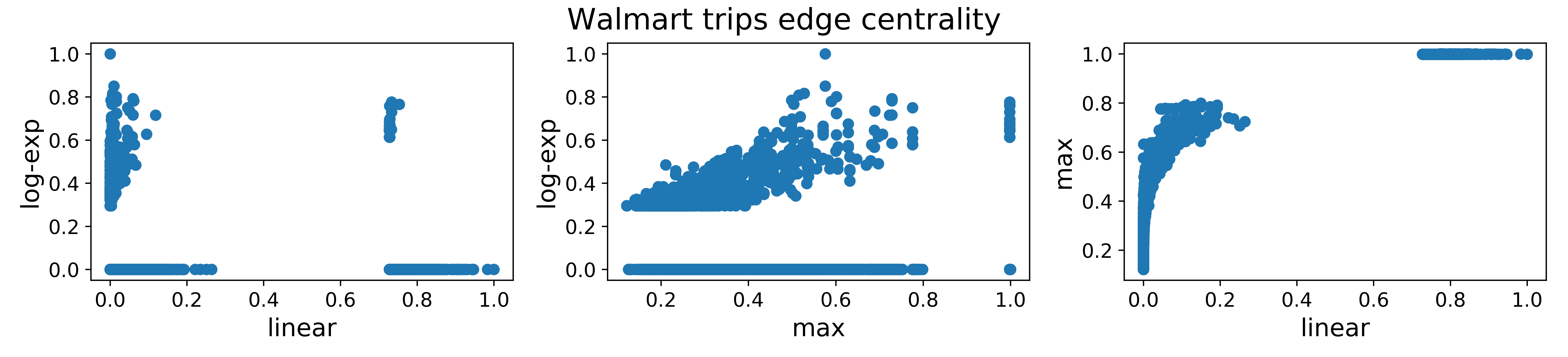}
    \caption{\textbf{Scatter plots of node and edge centralities.}\\
    Scatter plots comparing node and edge centrality scores obtained with the ``linear'', ``log-exp'' and ``max'' centrality models on the Math stackexchange co-tags and the Walmart trips hypergraphs. Each dot in the panels represents either a node or an hyperedge, with coordinates $(x,y)$ corresponding to the centrality value assigned to that node or hyperedge by two different models, as specified by the axis' labels.}
    \label{fig:scatter-edge-node-centrality}
\end{figure}

\begin{figure}[t]
    \centering
    \includegraphics[width=\textwidth]{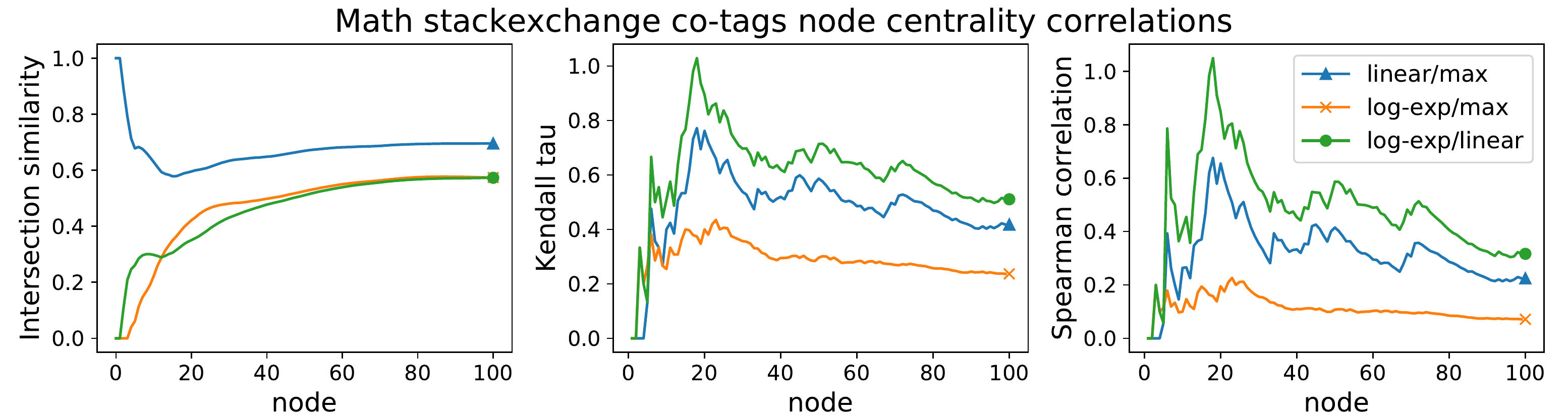}\\
    \includegraphics[width=\textwidth]{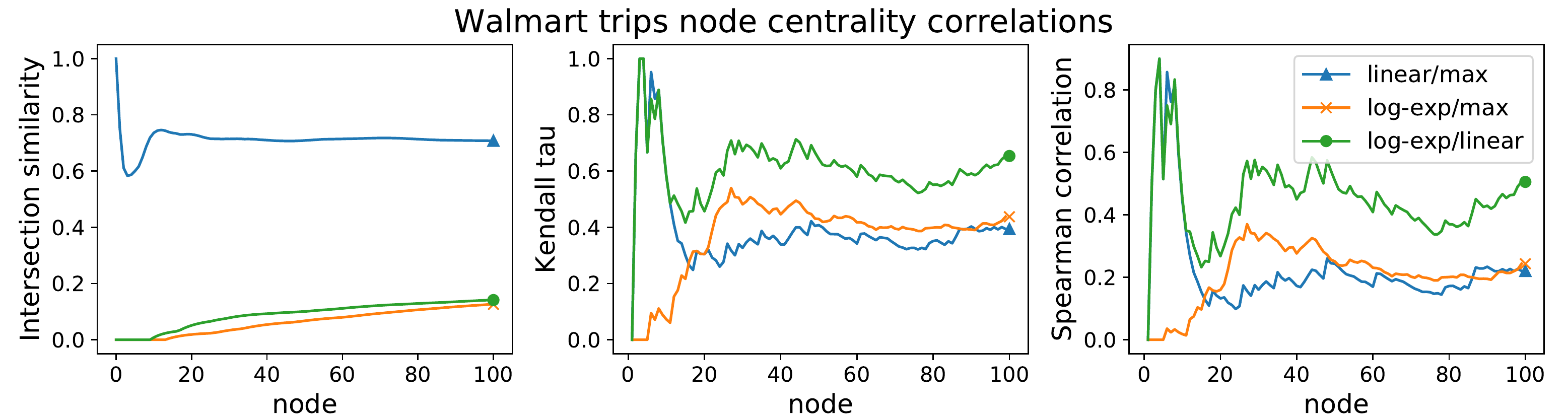}
    \caption{\textbf{Similarity between node and edge centralities. }\\ Intersection similarity, Kendall-$\tau$ correlation and Spearman correlation similarities between top $k$ nodes, for $k \in \{1,\dots,100\}$,  on the Math stackexchange co-tags and the Walmart trips hypergraph datasets. Nodes  are ranked according to  the ``linear'', ``log-exp'' and ``max'' centrality models and are paired as specified in the legend.}
    \label{fig:correlations}
\end{figure}

\begin{figure}
    \centering
    \includegraphics[width=\textwidth]{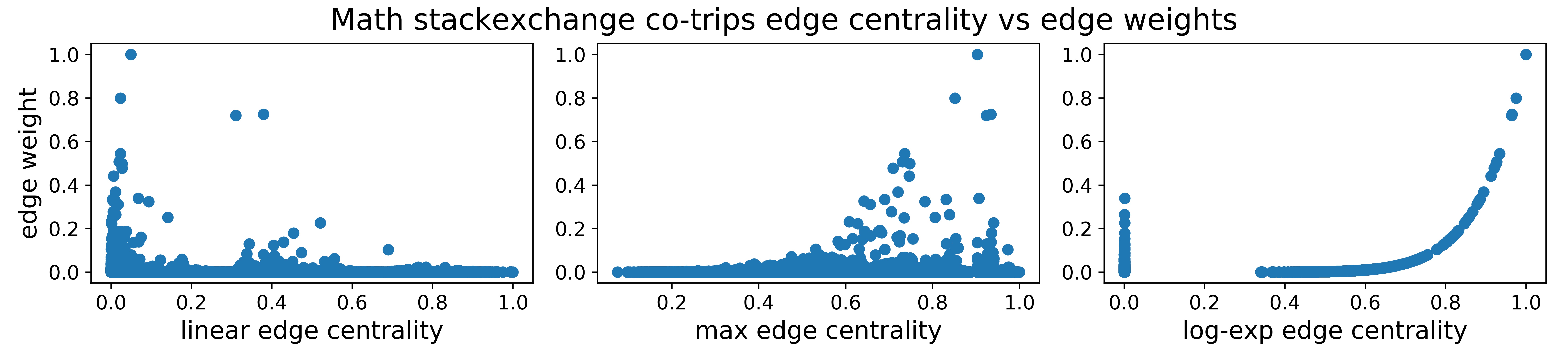}\\
    \includegraphics[width=\textwidth]{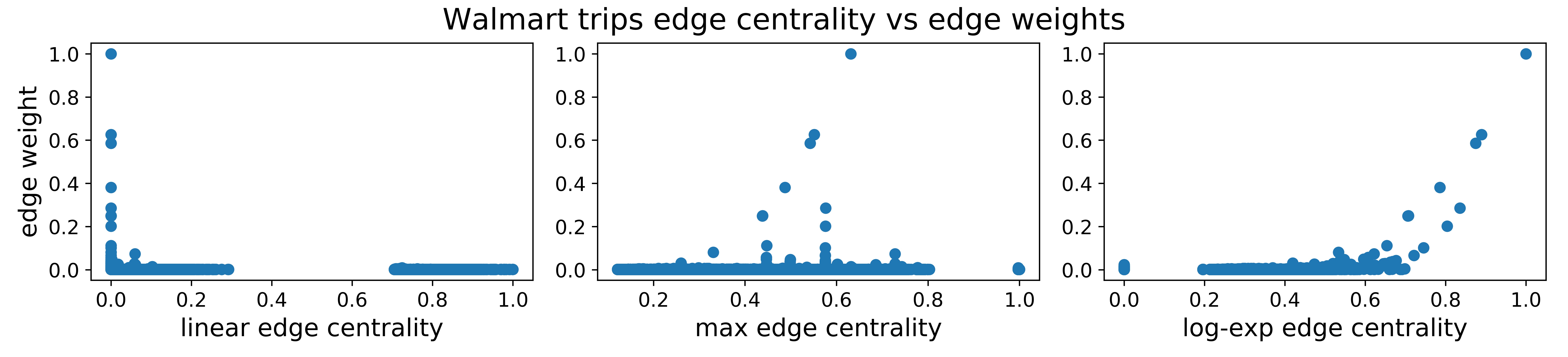}
    \caption{\textbf{Scatter plots of edge weights and centralities.}\\ Scatter plots comparing the edge weights against the edge centrality score computed with  the ``linear'', ``log-exp'' and ``max'' centrality models on the Math stackexchange co-tags and the Walmart trips hypergraphs. Each dot in the panels represents a hyperedge $e$, with coordinates $(x_e,y_e)$, where $x_e$ corresponds to the centrality value assigned to $e$ by one of the three models, whereas $y_e=w(e)$ is the weight of $e$ in the hypergraph.}
    \label{fig:scatter-edges}
\end{figure}

Scatter plots comparing node and edge centrality scores obtained with the ``linear'', ``log-exp'' and ``max'' centrality models on the Math stackexchange co-tags and the Walmart trips hypergraphs

\paragraph{Walmart trips.}
This is a transactional dataset that consists of a hypergraph describing customer trips at  Walmart:  hyperedges are sets of co-purchased products at Walmart, as released as part of the Kaggle competition in \cite{kaggle}. The hypergraph data is taken from \cite{Amburg-2020-categorical}. Products are assigned a label which points to one of ten broad departments in which each product appears on \texttt{walmart.com} (e.g., ``Clothing, Shoes, and Accessories''). There is also an additional ``Other'' class. The hyperedge weights are given by the number of times that a particular set of products appears in the list of co-purchased items. We summarize relevant statistics for this dataset in the list below:
\begin{itemize}[noitemsep]
    \item number of nodes: 88,860; number of hyperedges: 65,979;
    \item maximum edge weight: 679; mean / variance of edge weights: 1.06 / 15.24; 
    \item maximum edge size: 25; mean edge size 6.86.
\end{itemize}

\paragraph{Math stackexchange co-tags.}
This is a temporal higher-order network dataset from \cite{Benson-2018-simplicial}. Here we use the whole dataset ignoring the temporal component. The resulting dataset is a sequence of simplices where each simplex is a set of nodes. In this dataset, nodes are tags and simplices are the sets of tags applied to questions on \texttt{math.stackexchange.com}. We represent the dataset as a hypergraph with one hyperedge for each simplex. As before, the weight of each hyperedge is an integer number counting how many times that hyperedge appears in the data.  Some basic statistics of this dataset are:
\begin{itemize}[noitemsep]
    \item number of nodes: 1,629; number of hyperedges: 170,476;
    \item maximum edge weight: 16,230; mean / variance of edge weights: 4.82 / 9,430.71;
    \item maximum edge size: 5; mean edge size: 3.48.
\end{itemize}

\begin{table}[t]
    \centering
    \begin{tabular}{ccc}
    \toprule
    \textit{Linear} & \textit{Max} & \textit{Log-exp} \\
    \midrule
    Calculus & Calculus & Linear algebra \\
    Real analysis & Real analysis & Probability \\
    Integration & Linear algebra & Calculus \\
    Sequences and series & Probability & Real analysis \\
    Limits & Abstract algebra & Complex analysis \\
    Analysis & Integration & Algebra precalculus \\
    Derivatives & Sequences and series & General topology \\
    Linear algebra & Matrices & Differential equations \\
    Multivariable calculus & General topology & Combinatorics \\
    Definite integrals & Combinatorics & Geometry \\
    \bottomrule
\end{tabular}
    \caption{\textbf{Top ten nodes in the math stackexchange co-tags hypergraph.}\\
    Top ten nodes in the math stackexchange co-tags hypergraph, according to the ``linear'', ``log-exp'' and ``max'' centrality models, as specified in the top row of the table.}
    \label{tab:topnodes}
\end{table}

In Figure \ref{fig:scatter-edge-node-centrality} we scatter-plot the node and edge centrality obtained with these three models on the Walmart trip and the Math Stack-exchange co-tag hypergraphs. We  normalize the values of each centrality vector by dividing them entry-wise by their largest entry (so that their largest value is scaled to one). The figure compares the three centrality models in a pair-wise fashion and shows no apparent linear correlation between any pair of centralities, confirming  that different choices of the nonlinear functions lead to remarkably different centrality score assignments. 
This is further confirmed by  Figure \ref{fig:correlations}, where we plot the intersection similarity---left panel---as well as the Kendall-$\tau$ and the Spearman correlation coefficients---middle and right panels,
respectively---between the top $k$ nodes ranked by the linear model vs the ranking assigned to the same nodes by the other models, for $k$ which varies between $1$ and $100$. The intersection similarity \cite{fagin2003comparing} is a measure used to compare the top $k$ entries of two ranked lists $\ell^{1}$ and $\ell^{2}$ that 
may not contain the same elements.  It is defined as follows: let $\ell^{j}_k$ be the list of the top $k$ elements in $\ell^{j}$, for 
$j=1,2$.  Then, the  top $k$ intersection similarity between $\ell^{1}$ and $\ell^{2}$ is  
\begin{equation*}
{\rm isim}_k(\ell^{1},\ell^{2}) =1- \frac{1}{k}\sum_{t=1}^k \frac{|\ell^{1}_t\Delta\ell^{2}_t|}{2\, t},
\label{eq:isim}
\end{equation*}
where $|\ell^{1}_t\Delta\ell^{2}_t|$ denotes the number of elements in the  symmetric difference between $\ell^{1}_t$ and $\ell^{2}_t$. 
When the ordered sequences contained in $\ell^{1}$ and $\ell^{2}$ are completely different, then the intersection similarity between the two is minimum and it is equal to 0, whereas,  the intersection similarity between  $\ell^{1}$ and $\ell^{2}$  is equal to $1$ 
if and only if the two ordered sequences coincide.

As already observed for the case of the sunflower hypergraph, the `linear' and `log-exp'  models  may 
assign very different node centrality scores. This effect is also highlighted in  Table~\ref{tab:topnodes} where we report the top ten nodes with highest centrality for the three models for the math stackexchange datasets. For this dataset nodes are the tags that posts receive on the math-stackexchange website.

A similar comparison is illustrated in Figure \ref{fig:scatter-edges} where we scatter-plot the edge centrality of the three models with their edge weights. 
We see that, especially for the `linear' and `max' versions, larger edge weights 
do not correspond to greater importance in this spectral sense.

\section{Conclusion}\label{sec:summ}
Centrality measures give useful information in a range of
network science settings.
In the study of on-line human behaviour, such measures 
 are relevant to 
 targeted advertising \cite{LMGAOH13}, and to the
 detection of fake news generation \cite{PPC20}
  and other negative behaviours such as the spread of viruses and  cyberbullying
  \cite{AB18}.
  They have also proved useful 
in the physical world; for example in 
predicting (or vaccinating against) disease outbreaks
\cite{GMCCF14},
extracting biologically relevant features from 
neural connectivity data
\cite{CH09},
 and 
 quantifying the 
 attack robustness of power networks
\cite{CDM18}.
 
Taking the classical network view, where all nodes and pairwise connections play the same structural role,
typically requires us to trade-off some fine detail for the sake of elegance and simplicity.
By moving to higher order interactions, we are able to re-introduce some of this detail. Of course, 
in doing so we must understand the costs and benefits 
in terms of both the computational expense of the new  algorithms
and 
the ease with which results can be assimilated.
Our aim in this work was to show that 
the widely-used spectral approach to centrality measurement can be extended rigorously and 
at very little cost to the general hypergraph setting in order to quantify the importance of  both nodes and hyperedges.
As shown in Theorem~\ref{thm:existence_uniqueness}, there is a sound  
underlying theory behind the resulting constrained eigenvalue problems.
Further, as shown in Theorem~\ref{thm:convergence}, the measures can be computed by 
an efficient and globally convergent iteration 
(Algorithm~\ref{alg:npm})
that is built on matrix-vector products.

\section{Methods}\label{sec:methods}
This section provides the proofs of our three main theorems. 
\begin{proof}[Proof of Theorem \ref{thm:tensor-eig}]
First note that with $f=\text{id}$ and $g(\b x) =\b x^{1/(p+1)}$ and $N=I$, from \eqref{eq:NEP} we get $\lambda^{p+1} \b x^{p+1} = BW\b y$ and $\mu \b y = \psi(B^\top \varphi(\b x))$ which together imply that $\alpha \b x^{p+1} = BW\psi(B^\top \varphi(\b x))$ for some $\alpha > 0$. Now, as every edge $e$ contains exactly $k$ nodes, we can write $e=\{i_1,\dots,i_k\}$, yielding
$$
\psi(B^\top \varphi(\b x))_e = \exp\Big(\sum_{j\in e}\ln(x_j)\Big) = x_{i_1}\cdots x_{i_k}.
$$
Furthermore, for any $i_1\in V$ and any $\b y \in \RR^m$ we have $(BW\b y)_{i_1} = \sum_{e: i_1\in e}w(e)y_e$. Thus, if $\mathcal A$ is the adjacency tensor of $H$ (defined by $\mathcal A_{i_1,\dots,i_k}=w(e)$ if $e=\{i_1,\dots,i_k\}\in E$ and $\mathcal A_{i_1,\dots,i_k}=0$ otherwise) we get
$$
[BW\psi(B^\top \varphi(\b x))]_{i_1} = \sum_{e: i_1\in e}w(e)\psi(B^\top \varphi(\b x))_e = \sum_{i_2,\dots,i_k}\mathcal A_{i_1,\dots,i_k}x_{i_1}\cdots x_{i_k}.
$$
This shows that if $\b x$ solves \eqref{eq:NEP} then $\b x$ must be such that
\begin{equation*}\label{eq:step1}
    \sum_{i_2,\dots,i_k}\mathcal A_{i_1,i_2,\dots,i_k}x_{i_1}x_{i_2}\cdots x_{i_k} = \alpha  \, x_{i_1}^{\, p+1}  \, .
\end{equation*}
Finally, as $\b x$ is positive we can divide the previous identity by $x_{i_1}$, which reveals that $\b x$ solves the tensor eigenvalue problem in \eqref{eq:tensor_eig}. 
\end{proof}

\begin{proof}[Proof of Theorem \ref{thm:existence_uniqueness}]
The proof follows directly from the Perron--Frobenius theorem for multihomogeneous mappings \cite{gautier2019perron}. Below we highlight the main steps. Let $F:\RR^n\times \RR^m\to \RR^n\times \RR^m$ be the mapping
$$
F(\b x,\b y) = \big(\, g\big(BW f(\b y)\big),  \psi\big(B^\top N \varphi(\b x)\big)\, \big)\, .
$$
 A simple computation shows that $F$ is order-preserving and multihomogeneous with homogeneity matrix 
\begin{equation}\label{eq:M}
 M = \begin{pmatrix}
 0& |\alpha \beta| \\
 |\gamma\delta| & 0
 \end{pmatrix}
\end{equation} 
 and a solution of \eqref{eq:NEP} coincides with the multihomogeneous eigenvalue equation $F(\b x,\b y) = (\lambda \b x, \mu \b y)$. As the spectral radius of $M$ coincides with $\sqrt{|\alpha\beta\gamma\delta|}$, the thesis follows directly from \cite[Thm.\ 3.1]{gautier2019perron} under assumption P1. Assume now P2 holds. Since the mappings $f,g,\varphi,\psi$ act entry-wise, are homogeneous and order preserving, we have  that the graph $\mathcal G(F)$, as per \cite[Def.\ 5.1]{gautier2019perron},  coincides with the bipartite graph with adjacency matrix $
\begin{psmallmatrix}
0 & BW\\
B^\top N & 0
\end{psmallmatrix}
$. Thus, by \cite[Thm.\ 5.2]{gautier2019perron} there exist positive solutions to \eqref{eq:NEP}. Now let $(\b u, \b v)>0$ be any such solution and let $DF(\b u, \b v)$ be the Jacobian matrix of $F$ evaluated at $(\b u, \b v)$. Since $f,g,\varphi,\psi$ are positive and homogeneous, the nonzero pattern of $DF(\b u,\b v)$ coincides with the nonzero pattern of the matrix 
$
\begin{psmallmatrix}
0 & BW\\
B^\top N & 0
\end{psmallmatrix}
$.
Therefore, the thesis for assumption P2 eventually follows from \cite[Thm.\ 6.2]{gautier2019perron}.
\end{proof}

\begin{proof}[Proof of Theorem \ref{thm:convergence}]
This proof follows almost directly from the case of tensor eigenvectors, discussed in \cite[Thm.\ 3.3]{gautier2019unifying}. We tailor the main ideas of that argument to our hypergraph eigenvalue problem in \eqref{eq:NEP}.  Consider the mapping $G:\RR^n\times \RR^m\to \RR^n\times \RR^m$, defined by
$$
G(\b x,\b y) = \Big(\, \sqrt{\b x g\big(BW f(\b y)\big)} \, ,  \sqrt{\b y \psi\big(B^\top N \varphi(\b x)\big)}\, \Big),
$$
where all the operations are intended entrywise, and let $M$ be the homogeneity matrix defined in \eqref{eq:M}. It is not difficult to verify that $G$ is order-preserving and multihomogeneous, with homogeneity matrix $\tilde M = \frac 1 2 (M+I)$ and that, if $F$ is defined as in the  proof of Theorem \ref{thm:existence_uniqueness}, then $(\b x,\b y)$ is such that $F(\b x,\b y) =  (\lambda \b x ,\mu  \b y)$, for positive $\lambda$ and $\mu$, if and only if $G(\b x,\b y) = (\tilde \lambda\b x, \tilde \mu  \b y)$, with $\tilde \lambda,\tilde \mu>0$. Moreover, a direct computation shows that  the Jacobian matrices  $DF(\b x,\b y)$ and $DG(\b x,\b y)$ of $F$ and $G$, respectively, are such that 
\begin{equation}\label{eq:DG}
    DG(\b x,\b y) = \frac 1 2 \Diag\big(G(\b x,\b y)\big)^{-1/2}\Big(\Diag\big(F(\b x,\b y)\big)+\Diag\big((\b x,\b y)\big)DF(\b x,\b y)\Big),
\end{equation}
where $\Diag(\b v)$ is the diagonal matrix with diagonal entries given by the elements of $\b v$. As observed in the proof of Theorem \ref{thm:existence_uniqueness},  for a positive vector $(\b x,\b y)$ the matrix $DF(\b x,\b y)$ is irreducible. Thus, from \eqref{eq:DG}, $DG(\b x,\b y)$ is primitive and the thesis eventually follows from  \cite[Thm.\ 7.1]{gautier2019perron}. 
\end{proof}

\section*{Data and code availability}
All data  and code used in this work is publicly available under CCBY 4.0 licence via the online repository \url{https://github.com/ftudisco/node-edge-hypergraph-centrality}

\section*{Funding}   
DJH was supported by EPSRC Programme Grant EP/P020720/1.


\end{document}